\documentclass[11pt]{amsart}
\usepackage{graphicx,amssymb,amsmath,amsthm}
\usepackage{enumerate}
\usepackage{comment,cite,color}
\usepackage{cite,color}
\usepackage{algorithmic}
\usepackage{mathrsfs}
\usepackage[ruled]{algorithm}
\usepackage{epsfig}
\usepackage{enumitem}

\newcommand{\PP}{\mathbb{P}}

\newcommand{\innerp}[1]{\langle {#1} \rangle}

\newcommand{\abs}[1]{\lvert#1\rvert}

\newcommand{\argmin}[1]{\mathop{\rm argmin}\limits_{#1}}

\newcommand{\E}{\mathbb E}

\newcommand{\R}{{\mathbb R}}

\renewcommand{\eqref}[1]{(\ref{#1})}

\newcommand{\mhsp}{\hspace{2em}}

\newtheorem{prop}{Proposition}[section]
\newtheorem{lem}[prop]{Lemma}

\newtheorem{defi}{Definition}[section]

\newtheorem{theo}[prop]{Theorem}
\newtheorem{remark}[prop]{Remark}

\usepackage[top=1in,bottom=1in,left=1in,right=1in]{geometry}
\date{}

\begin{document}
\bibliographystyle{plain}
\title[A strong RIP  with an application to phaseless compressed sensing]{A strong restricted isometry property, with an application to phaseless compressed sensing.}

\author{Vladislav Voroninski}
\thanks{
       Zhiqiang Xu  is supported by the National Natural Science Foundation of China (11171336 and 11331012) and by the Funds for Creative
       Research Groups of China (Grant No. 11021101).
       }
\address{Massachusetts Institute of Technology. Department of Mathematics. 77 Massachusetts Avenue
Cambridge, MA 02139-4307.}
\email{vvlad@math.mit.edu}

\author{Zhiqiang Xu}
\address{LSEC, Inst.~Comp.~Math., Academy of
Mathematics and System Science,  Chinese Academy of Sciences, Beijing, 100091, China}
\email{xuzq@lsec.cc.ac.cn}

\maketitle

\begin{abstract}
The many variants of the restricted isometry property (RIP) have proven to be crucial theoretical tools in the fields of compressed sensing and  matrix completion. The study of extending compressed sensing to accommodate phaseless measurements naturally motivates a strong notion of restricted isometry property (SRIP), which we develop in this paper. We show that if $A \in \mathbb{R}^{m\times n}$ satisfies SRIP and phaseless measurements $|Ax_0| = b$ are observed about a $k$-sparse signal $x_0 \in \mathbb{R}^n$, then minimizing the $\ell_1$ norm subject to $ |Ax| = b $ recovers $x_0$ up to multiplication by a global sign. Moreover, we establish that the SRIP holds for the random Gaussian matrices typically used for standard compressed sensing, implying that phaseless compressed sensing is possible from $O(k \log (n/k))$ measurements with these matrices via $\ell_1$ minimization over $|Ax| = b$. Our analysis also yields an erasure robust version of the Johnson-Lindenstrauss Lemma.
\end{abstract}

\vspace{0.4cm}


\section{Introduction}
\setcounter{equation}{0}

The restricted isometry property (RIP), first introduced by Cand\`es  and Tao \cite{CanTao05},
is one of the most commonly used tools in the study of sparse/low rank signal recovery problem.  The RIP also has some connections to the Johnson-Lindenstrauss lemma and its study has lead to new results about the lemma \cite{RIPJL, DeVoreRIP}.
The aim of this paper is to present  a strong restricted isometry property  which naturally occurs   when considering phaseless compressed sensing.

Given a vector $x_0 \in \R^n$ and a collection of phaseless measurements $b_j= \abs{\innerp{a_j, x_0}}, j = 1,2,\ldots, m$ where $a_j\in \R^n$, {\em phase retrieval} consists of recovering $x_0$ up to a global sign,
which has attracted much attention in recent years (cf. \cite{BCE06, BBCE07, CSV12, CESV12}).
 It has been established that the PhaseLift algorithm allows for stable and efficient phase retrieval of an arbitrary signal $x_0 \in \mathbb{C}^n$ from $O(n)$ measurements via semi-definite programming \cite{CSV12}.

In practice, signals of interest are often sparse in some basis and in particular this occurs in some regimes of X-ray crystallography. It is natural to exploit this sparsity structure to minimize the number of measurements needed for recovery since measurement acquisition is expensive and can destroy the sample at hand. We define {\em phaseless compressed sensing} (PCS) as the problem of recovering a sparse signal from few such phaseless measurements. It was shown in \cite{LiV} that a $k$-sparse signal $x_0$ can be recovered from $O(k^2 \log n)$ phaseless measurements via convex programming. Surprisingly, and in contrast to the case of compressed sensing from linear measurements, it was also established in \cite{LiV} that the natural information theoretic lower-bound of $O(k \log n)$ measurements cannot be achieved using a naive semi-definite programming relaxation.

Meanwhile, phaseless measurements are generically injective modulo phase over $k$-sparse signals as soon as the over-sampling factor is $2$ \cite{wangxu, LiV}. Thus, the combinatorially hard problem of finding
\begin{equation}\label{eq:l0}
\argmin{x\in \R^n}\{\|x\|_0 \; \mbox{  subject to } \;  \abs{\innerp{a_j,x}}=\abs{\innerp{a_j,x_0}}, j=1,\ldots,m\}
\end{equation}
 yields $x_0$ modulo phase for $m \geq 2k-1$ and $a_j$  generic (see \cite{wangxu}). Moreover, $O(k \log( n/k))$ random Gaussian phaseless measurements separate signals well \cite{Eldar}. While minimizing sparsity in the $O(k \log (n/k))$ measurement regime is not clearly amenable to efficient algorithmic recovery,
  numerical experiments show that  using a convex relaxation of the $\ell_0$ ``norm" is often exact \cite{algorithm1,algorithm2,algorithm3}. We study here such a relaxation:
\begin{equation}\label{eq:l1}
\hat{x}:=\argmin{x \in \R^n} \{\|x\|_1 \; \mbox{  subject to } \; \abs{\innerp{a_j,x}} = \abs{\innerp{a_j,x_0}}, j=1,\ldots,m \}.
\end{equation}
In this paper, we focus on the case where $x_0\in \R^n$ is a $k$-sparse real signal and $A=[a_1,\ldots,a_m]^\top \in  \R^{m\times n}$.
 Particularly, we  show that (\ref{eq:l0}) is equivalent to (\ref{eq:l1}) provided that the matrix $A$ satisfies the strong restricted isometry property.
We furthermore show that a random $m \times n$ Gaussian matrix with $m=O(k \log( n/k))$ satisfies the strong restricted isometry property of order $k$ with
high probability. And hence, the $\ell_1$ relaxation is exact with high probability with $O(k \log (n/k))$ Gaussian measurements, just as in traditional compressed sensing. While the constraint set of our relaxation is non-convex, this formulation is more amenable to algorithmic recovery, and has been studied in \cite{algorithm1,algorithm2,algorithm3} with demonstrated empirical success of the proposed projection algorithms.

The paper is organized as follows.
In Section 2, we introduce the definition of strong restricted isometry property (SRIP) and
show that if a matrix $A$ satisfies SRIP, then (\ref{eq:l0}) is equivalent to (\ref{eq:l1}). Then, we show that $m \times n$ random Gaussian matrices typically used for linear compressed sensing, also satisfy the SRIP with high probability, which establishes the exactness of the $\ell_1$ relaxation for $m = O(k \log n/k)$ phaseless  measurements.
 Section 3 provides some technical necessities. We present a strong version of the
 concentration of measure inequality in Section 4
 which plays an important role in proving the main  results of Section 2.
 Using this inequality, we derive a strong J-L lemma which
 states that we get a dimensionality reduction of a set of points in Euclidean space with some distance distortion such that erasing a positive fraction of the coordinates of the reduction maintains to a certain degree the approximate distance preserving J-L property. This can be interpreted as an erasure robust version of the J-L Lemma. We present the proofs of the main results in Section 5 and finally give some future research directions in Section 6.

\section{Main results}

We say that a matrix $A=[a_1,\ldots,a_m]^\top \in  \R^{m\times n}$ satisfies the Restricted Isometry Property (RIP) of order $k$ and constant $\delta_k \in [0,1)$ if
\begin{equation}\label{eq:con}
(1-\delta_k) \|x\|_2^2 \leq \|Ax\|_2^2 \leq (1+\delta_k) \|x\|_2^2
\end{equation}
holds for all $k$-sparse signals $x\in \R^n$ (see \cite{CanTao05}). Matrices which satisfy the RIP property play an important role in analyzing the performance of $\ell_1$ minimization.
Particularly, in \cite{candes,tao}, it is shown that,  with certain RIP constants $\delta_k$,
 for any $k$-sparse vector $x_0\in \R^n$
\begin{equation*}
   \argmin{x\in \R^m}\{\|x\|_1:~A x = Ax_0\} = \{x_0\}.
\end{equation*}

For convenience, let $[m]:=\{1,\ldots,m\}$. Now we introduce the definition of SRIP:
\begin{defi}
We say the matrix $A\in \R^{m\times n}$  satisfies the Strong Restricted Isometry Property (SRIP) of order $k$ and levels $\theta_-, \theta_+
\in (0,2)$
if
\[
\theta_-\|x\|_2^2 \leq \min_{I\subseteq [m], \abs{I} \geq m/2}\|A_Ix\|_2^2 \leq \max_{I\subseteq [m], \abs{I} \geq m/2} \|A_I x\|_2^2 \leq \theta_+\|x\|_2^2
\]
holds for all $k$-sparse signals $x\in \R^n$.  Here $A_{I}:=[a_j:j\in I]^\top$ denotes the sub-matrix of $A$ where only rows with
indices in $I$ are kept.
\end{defi}
It is well-known that an $m\times n$ Gaussian matrix with $m=O(k\log(n/k))$ satisfies the RIP property of order $k$ with high probability. We establish below that the SRIP also holds for random Gaussian matrices of the same size with high probability, for constants $\theta_-, \theta_+$ which are necessarily bounded below by some non-zero universal constant. We have the following result:
\begin{theo}\label{th:SRIP}
Suppose that $t>1$ and that $A\in \R^{m\times n}$ is a random Gaussian matrix with $m=O(tk\log (n/k))$. Then
there exist constants $\theta_-, \theta_+$ with $0<\theta_-<\theta_+<2$, independent of t, such that
$A$ satisfies SRIP of order $t\cdot k$ and levels $\theta_-, \theta_+$ with probability $1-\exp(-cm/2)$, where $c>0$ is an absolute constant
.
\end{theo}
The proof of Theorem \ref{th:SRIP} is inspired by the method introduced in \cite{DeVoreRIP},
 which combines  a strong concentration of measure inequality and a simple covering argument.
The next theorem shows that one can use $\ell_1$ minimization to recover sparse signals from the phaseless measurements provided $A$
satisfies SRIP.
\begin{theo}  \label{theo-3.3}
 Assume that $A\in \R^{m\times n}$ satisfies the Strong RIP of order $t\cdot k$ and levels $\theta_-,\theta_+$ with $t\geq \max\{\frac{1}{2\theta_--\theta_-^2}, \frac{1}{2\theta_+-\theta_+^2}\}$.
Then for any $k$-sparse signal $x_0\in \R^n$ we have
\begin{equation} \label{3.4}
   \argmin{x\in \R^n}\{\|x\|_1:~\abs{A x} = \abs{A x_0}\} = \{\pm x_0\},
\end{equation}
where $\abs{Ax}:=[\abs{\innerp{a_j,x}}:j\in [m]]$ and $[m]:=\{1,\ldots,m\}$.
\end{theo}

\begin{remark}
Theorem \ref{th:SRIP} shows that Gaussian random matrixes  satisfy SRIP with high probability. Another popular measurement ensemble in compressed sensing is the Bernoulli ensemble, which is defined as
$$
\PP(a_{j,i}=1/\sqrt{m})\,=\,\PP(a_{j,i}=-1/\sqrt{m})\,=\,1/2.
$$
For any matrix A in this ensemble, we set
$$
I_0:=\{j\in [m] :  a_{j,1}=a_{j,2}\},
$$
i.e.,  the first two entries of $a_j$ are same provided $j\in I_0$.
Then either $\abs{I_0}\geq m/2$ or $\abs{I_0^c}\geq m/2$ holds. Without loss of generality, we assume that $\abs{I_0}\geq m/2$. Then
a simple observation is that the first two columns of the matrix $A_{I_0}:=[a_j:j\in I_0]^\top$ are linearly dependent, which implies that $A$
does not satisfy SRIP of order $k\geq 2$.
\end{remark}

\begin{remark}
Restricted Isometry Properties of some $m \times n$ matrix $A$ can be interpreted as controlling the singular values of various submatrices of $A$. For instance, establishing that A has the RIP of order $k$ and level $\delta$, is equivalent to saying that the singular values of any $m \times k$ submatrix of A lie in a $\delta$-neighborhood of 1. From this perspective, the SRIP with these parameters demands that the same is true of any $m' \times k$ submatrix of A for which $m' \geq m/2$, which in turn means that any $m' \times n$ submatrix of A, with $m' \geq m/2$, satisfies RIP with the aforementioned parameters. This can be interpreted as an erasure robust-property. D. Mixon and A. Bandiera have studied Numerically Erasure Robust Frames \cite{NERF}, which  instead have the property that singular values of any $m' \times n$ submatrix of A, with $m' > m/2$ are in a $\delta$ neighborhood of 1. Our results are complementary. For instance, the RIP property is robust to arbitrarily large erasures, while this is unknown for NERFs.
\end{remark}

\section{Preliminaries}

{\bf Restricted isometry property condition.} As mentioned before, when the matrix $A$ satisfies a certain RIP condition,
one can use $\ell_1$ minimization to recover $k$-sparse signals \cite{candes,tao}.
The result in  \cite{CaiZhang13} improves the sufficient RIP condition for relaxation exactness and shows that
for any $k$-sparse vector $x_0\in \R^n$
\begin{equation*}
   \argmin{x\in \R^n}\{\|x\|_1:~A x = Ax_0\} = \{x_0\}
\end{equation*}
provided $A$ satisfies the RIP of order $t\cdot k$ and $\delta_{t\cdot k}<\sqrt{1-\frac{1}{t}}$ where $t>4/3$.

{\bf The concentration of measure inequalities.}
Suppose that $A\in \R^{m\times n}$ is a Gaussian matrix whose entries $a_{jk}$ are
independent realizations of Gaussian random variables $a_{jk}\sim {\mathcal N}(0,1/m)$. Suppose that $\bar{x}\in \R^n$ is a fixed vector. Then
for all $\epsilon\in (0,1)$ we have
\begin{equation}\label{eq:conc}
{\mathbb P}(\abs{\|A\bar{x}\|_2^2-\|\bar{x}\|_2^2}\geq \epsilon \|{\bar x}\|_2^2)\leq 2\exp(-m(\epsilon^2/4-\epsilon^3/6)).
\end{equation}
The authors of \cite{DeVoreRIP} used (\ref{eq:conc}) to present a simple proof that the random matrices at hand satisfy RIP.
This inequality (\ref{eq:conc}) also follows by the concentration of
measure of Gaussian space:
\begin{theo}(\cite{probability})\label{th:Lip}
Let $f: \mathbb{R}^m \mapsto \R$ be Lipschitz with constant $1$: $|f(x)-f(y)| \leq \|x-y\|_2$. Then for $X_1,\ldots X_m$ iid standard normal random variables  we have
\[
\PP\left[ |f(X_1,\ldots X_m) - \E\left[ f(X_1,\ldots X_m) \right] | \geq t \right] \leq 2e^{-t^2/2}, \label{eq: GC}
\]
where $t\geq 0$.
\end{theo}

{\bf The expectation of the $k$th smallest random variable.}
Follow the definition in \cite{min}, we say that a random variable $\xi$ satisfies the $(\alpha,\beta)$-condition if
$$
{\mathbb P}(\abs{\xi}\leq t)\leq \alpha t\qquad \mbox{ for every } t\geq 0
$$
and
$$
{\mathbb P}(\abs{\xi}>t)\leq \exp(-\beta t) \qquad \mbox{ for every } t\geq 0,
$$
where $\alpha>0, \beta>0$ are parameters. Then the following theorem presents a lower bound
 for the expectation of the $k$th smallest order statistic of $m$ such independent random variables.
\begin{theo}(\cite{min})\label{th:min}
Let $\alpha>0, \beta>0$.  Let $\xi_1,\ldots,\xi_m$
be independent random variables satisfying the $(\alpha,\beta)$-condition. Then
$$
c_{\alpha}\max_{1\leq j\leq k}\frac{k+1-j}{m-j+1}\leq \E\abs{\xi}_{(k)}
$$
where $c_{\alpha}=\frac{1}{2e\alpha}(1-\frac{1}{4\sqrt{\pi}})$ and $\abs{\xi}_{(k)}$ is the $k$th smallest order statistic, i.e.,
$\abs{\xi}_{(1)}\leq \cdots \leq\abs{\xi}_{(m)} $.
\end{theo}

{\bf Johnson-Lindenstrauss Lemma.} The J-L lemma, which has proven to be a useful tool in dimensionality reduction, follows easily from (\ref{eq:conc}) :
\begin{lem}[Johnson-Lindenstrauss] Let $\epsilon\in (0,1)$ be given. For every set $Q$ of $\abs{Q}$ points in
$\R^n$, if $m$ is a positive integer such that $m>O(\ln(\abs{Q})/\epsilon^2)$, then there exists a
Lipschitz mapping $f:\R^n\rightarrow \R^m$ such that
$$
(1-\epsilon)\|u-v\|_2^2\leq \|f(u)-f(v)\|_2^2 \leq (1+\epsilon)\|u-v\|_2^2
$$
for all $u,v\in Q$.

\end{lem}

\section{ The strong concentration of measure inequality and Johnson-Lindenstrauss Lemma }

In this section, we extend the concentration inequality  (\ref{eq:conc}) to a stronger version which plays an important
role in our proof of the main results.

 We first introduce some notation. For some $x \in \mathbb{R}^m$, set $\abs{x}:=(\abs{x_1},\ldots,\abs{x_m})$ and  $x^<:=(x_{(1)},\ldots,x_{(m)})\in \R^m$ where $x_{(1)}\leq \cdots \leq x_{(m)}$. Recall that $[m]:=\{1,\ldots,m\}$. Now define
$$
F(x):=\sqrt{\sum_{j=1}^{\lceil m/2\rceil}\abs{x}_{(j)}^2},
$$
where $\abs{x}_{(1)}\leq \cdots \leq \abs{x}_{(m)}$.
 Then we have
\begin{lem}\label{le:Lip}
The function $F: \mathbb{R}^m \mapsto \R$ is Lipschitz with constant $1$, i.e.,
$$
\abs{F(x)-F(y)}\leq \|x-y\|_2
$$
for any $x, y\in \R^m$.
\end{lem}
\begin{proof}
Consider
\begin{eqnarray*}
\|\abs{x}^<-\abs{y}^<\|_2^2&=&\|\abs{x}^<\|_2^2+\|\abs{y}^<\|_2^2-2\sum_{j=1}^m(\abs{x}^<)_j(\abs{y}^<)_j\\
&=&\|x\|_2^2+\|y\|_2^2-2\sum_{j=1}^m(\abs{x}^<)_j(\abs{y}^<)_j\\
&\leq &\|x\|_2^2+\|y\|_2^2-2\sum_{j=1}^m\abs{x}_j\abs{y}_j\\
&=&\|\abs{x}-\abs{y}\|_2^2.
\end{eqnarray*}
Here, in the third inequality, we use the rearrangement inequality. Then we have
\begin{eqnarray*}
\abs{F(x)-F(y)} &\leq& \|\abs{x}^{<}-\abs{y}^<\|_2\\
&\leq &\|\abs{x}-\abs{y}\|_2\\
&\leq &\|x-y\|_2.
\end{eqnarray*}
\end{proof}

\begin{lem}\label{le:mun}
  We assume that $X_1,\ldots,X_m$ are i.i.d. ${\mathcal N}(0,1)$ and set
$$
\mu_m:=\frac{1}{\sqrt{m}}\E\left[\sqrt{\sum_{j=1}^{\lceil m/2\rceil }\abs{X}_{(j)}^2}\right],
$$
where $\abs{X}_{(1)}\leq \abs{X}_{(2)}\leq \cdots \leq \abs{X}_{(m)}$. Then for any $m\geq 1$, we have
$$
\mu_m\geq \nu_0
$$
where $\nu_0:= \frac{1}{32e } \cdot  \sqrt{\frac{\pi}{2}}\cdot \left(1-\frac{1}{4\sqrt{\pi}} \right)\thickapprox 0.0124 $.
\end{lem}
\begin{proof}

We first consider the case where $m\leq 8$. Then Theorem \ref{th:min} implies that
\begin{eqnarray}
\mu_m &\geq& \sqrt{\frac{\lceil m/2\rceil}{m}}\E \abs{X}_{(1)}\geq \sqrt{\frac{\pi}{2}}\frac{1}{2e}\left(1-\frac{1}{4\sqrt{\pi}}\right)\frac{1}{m}\nonumber\\
 &\geq   & \frac{1}{64} \sqrt{\frac{\pi}{2}} \frac{1}{e}\left(1-\frac{1}{4\sqrt{\pi}}\right).\label{eq:v0}
 \end{eqnarray}
We next only consider the case where $m> 8$.
By applying Theorem \ref{th:min} to standard Gaussian rvs, for which $\alpha = \sqrt{2/\pi}$, we obtain that
$$
\E\abs{X}_{(k)} \geq \frac{1}{2e } \cdot  \sqrt{\frac{\pi}{2}}\cdot \left(1-\frac{1}{4\sqrt{\pi}} \right)\cdot \left(\max_{1\leq j\leq k} \frac{k+1-j}{m-j+1}\right).
$$
 We take $k=\lfloor m/4 \rfloor$ and obtain that
\begin{eqnarray*}
\E\abs{X}_{(\lfloor m/4 \rfloor)} &\geq & \frac{1}{2e } \cdot  \sqrt{\frac{\pi}{2}}\cdot \left(1-\frac{1}{4\sqrt{\pi}} \right) \left(\max_{1\leq j\leq \lfloor m/4 \rfloor} \frac{\lfloor m/4 \rfloor+1-j}{m-j+1}\right) \\
&\geq &  \frac{1}{2e } \cdot  \sqrt{\frac{\pi}{2}}\cdot \left(1-\frac{1}{4\sqrt{\pi}} \right) \left(\frac{  \lfloor m/4 \rfloor }{m}\right) \\
&\geq & \frac{1}{2e } \cdot  \sqrt{\frac{\pi}{2}}\cdot \left(1-\frac{1}{4\sqrt{\pi}} \right) \left(\frac{1}{4}-\frac{1}{m}\right)\\
&\geq & \frac{1}{16e } \cdot  \sqrt{\frac{\pi}{2}}\cdot \left(1-\frac{1}{4\sqrt{\pi}} \right).
\end{eqnarray*}
Then
\begin{eqnarray}
\mu_m=\frac{1}{\sqrt{m}}\E\left[\sqrt{\sum_{j=1}^{\lceil m/2\rceil }\abs{X}_{(j)}^2}\right]&\geq& \frac{1}{\sqrt{m}}\E\left[\sqrt{\sum_{j=\lfloor m/4\rfloor}^{\lceil m/2\rceil }\abs{X}_{(j)}^2}\right]\nonumber\\
&\geq& \frac{1}{\sqrt{m}}\sqrt{\lceil m/2\rceil-\lfloor m/4\rfloor} \E\abs{X}_{(\lfloor m/4 \rfloor)}\nonumber\\
&\geq & \frac{1}{2}  \E\abs{X}_{(\lfloor m/4 \rfloor)}\nonumber\\
&\geq & \frac{1}{32e } \cdot  \sqrt{\frac{\pi}{2}}\cdot \left(1-\frac{1}{4\sqrt{\pi}} \right).\label{eq:v2}
\end{eqnarray}
Combining (\ref{eq:v0}) and (\ref{eq:v2}), we obtain that
$$
\mu_m\geq \min\left\{\frac{1}{32e } \cdot  \sqrt{\frac{\pi}{2}}\cdot \left(1-\frac{1}{4\sqrt{\pi}} \right), \frac{1}{64} \cdot \sqrt{\frac{\pi}{2}} \cdot \frac{1}{e}\cdot \left(1-\frac{1}{4\sqrt{\pi}}\right) \right\}=\nu_0
$$
for any $m\geq 1$.
\end{proof}

Combining the results above, we arrive at
\begin{lem}\label{le:24}
Under the conditions of  Lemma \ref{le:mun} we have
\[
\PP\left[ (\mu_m-t)^2\leq \frac{1}{m}\sum_{i=1}^{\lceil m/2\rceil}(|X|_{(i)})^2   \leq (\mu_m + t)^2 \right] \geq 1- 2\exp{(-mt^2/2)}
\]
where $0\leq t\leq \mu_m$.
\end{lem}
\begin{proof}
Combining Lemma \ref{le:Lip} and Theorem \ref{th:Lip}, we obtain that
\[
\PP\left[ \left| \sqrt{\sum_{i=1}^{\lceil m/2\rceil}(|X|_{(i)})^2} - \sqrt{m}\mu_m  \right| \geq \tau \right] \leq 2e^{-\tau^2/2},
\]
where $\tau\geq 0$. Taking $\tau:=\sqrt{m}t$, we arrive at
\[
\PP\left[ -t\leq \sqrt{\frac{1}{m}\sum_{i=1}^{\lceil m/2\rceil }(|X|_{(i)})^2} - \mu_m   \leq t \right] \geq 1- 2e^{-mt^2/2},
\]
which implies the conclusion.
\end{proof}

We next state the  strong concentration of measure inequalities:
\begin{lem}\label{le:SJL}
Suppose that $A\in \R^{m\times n}$ is a Gaussian matrix whose entries $a_{jk}$ are
independent realizations of Gaussian random variables $a_{jk}\sim {\mathcal N}(0,1/m)$. Suppose that $\bar{x}\in \R^n$ is a fixed vector.
Then there exist absolute constants $0<c_-\leq  c_+<2$ so that
\[
c_-\|\bar{x}\|_2^2 \leq \min_{I\subseteq [m], \abs{I} \geq m/2}\|A_I\bar{x}\|_2^2 \leq \max_{I\subseteq [m], \abs{I} \geq m/2} \|A_I \bar{x}\|_2^2 \leq c_+\|\bar{x}\|_2^2
\]
holds with probability $1- 2e^{-cm}$ where $c>0$
 is an absolute constant.
\end{lem}
\begin{proof}
Without loss of generality, we assume that $\|\bar{x}\|_2=1$.
Set $y:=A\bar{x}$ and $y_I:=A_I\bar{x}$ where $I\subseteq [m]$. Then the entries of $y$ are independent realizations of Gaussian random variables $y_j\sim {\mathcal N}(0,1/m)$. Based on  Lemma \ref{le:mun}, $\mu_m \geq \nu_0$  for any $m\geq 1$.
Taking $t=\nu_0/2$ in Lemma \ref{le:24}, we have
\[
\PP\left[  \min_{I\subseteq [m], \abs{I} \geq m/2}\|y_I\|_2^2   \geq c_- \right] \geq 1- 2e^{-\nu_0^2m/8},
\]
where $c_-:=\nu_0^2/4$.
Note that
$$
\max_{I\subseteq [m], \abs{I} \geq m/2} \|A_I \bar{x}\|_2^2  \leq \|A\bar{x}\|_2^2.
$$
And hence, the upper bound follows from the proof of the classical concentration of measure inequalities (\ref{eq:conc}).
\end{proof}

Combining Lemma \ref{le:SJL} and a standard probability argument of J-L lemma, we can obtain
the erasure-robust version of the J-L Lemma, which can be interpreted as a J-L map that has robustness to corrupted measurements.
\begin{lem}\label{le:SJL2}
 For every set $Q$ of $\abs{Q}$ points in
$\R^n$, there exists a Lipschitz map $f: \mathbb{R}^n \mapsto \mathbb{R}^m$ with $m = O(\log \abs{Q})$ and absolute constants $0<c_-<c_+<2$ such that
\[
c_- \|u-v\|_2^2 \leq \|f_I(u)-f_I(v)\|_2^2 \leq c_+ \|u-v\|_2^2
\]
holds   for all  $u,v\in  Q$ and all $I\subset [m]$ with $\abs{I}\geq m/2$. Here $f_I(u)$ denotes the sub-vector of $f(u)$ only the
entries with the indices in $I$ are kept.
\end{lem}
\begin{remark}
Based on the argument of Lemma \ref{le:SJL}, one can observe that $c_-\thickapprox 10^{-5}$ and also take $c_+=1+\epsilon$ for arbitrary fixed
$\epsilon>0$. Numerical experiments show that Lemma \ref{le:SJL2} still holds if one takes $c_-\thickapprox 0.1$. And hence, it would
be interesting to improve the constant $c_-$ in Lemma \ref{le:SJL2}.
\end{remark}

\section{Proofs of Theorem \ref{th:SRIP} and \ref{theo-3.3}}
\begin{proof}[Proof of Theorem \ref{th:SRIP} ]
First, note that $\|A_Ix\|_2^2 \leq \|Ax\|_2^2$ for any $I\subset [m]$. And hence, according to RIP theory, there exists $1<\theta_+<2$ so that
$$
\max_{I\subset [m]}\|A_Ix\|_2^2\leq \theta_+\|x\|_2^2.
$$
holds for all $tk$-sparse $x\in \R^n$ with probability $1-\exp(-cm)$ provided $m=O(tk\log(n/k))$.
We next consider the lower bound.
Without loss of generality, we assume that $\|x\|_2=1$. So, to this end, we need to prove that there exists $\theta_-\in (0,1)$ so that
\[
\theta_-\leq \min_{I\subseteq [m], \abs{I} \geq m/2}\|A_Ix\|_2^2
\]
holds for all $x\in \R^n$ satisfying $\|x\|_2=1$ and $\abs{{{\rm supp}(x)}}\leq tk $.  We first consider a fixed support $\Omega\subset [m]$ with $\abs{\Omega}=tk$. Suppose that ${\mathcal N}_\epsilon $ is an $\epsilon$-net of
$$
\{x\in \R^n: \|x\|_2=1, {\rm supp}(x)\subset \Omega\}.
$$
Note that $\abs{{\mathcal N}_\epsilon}\leq (12/\epsilon)^{tk}$. Then, according to Lemma \ref{le:SJL}, there exists $c_->0$ such that
$$
c_-\leq \min_{I\subseteq [m], \abs{I} \geq m/2}\|A_I{ x}\|_2^2
$$
holds for all ${ x}\in {\mathcal N}_\epsilon$ with probability at least $1-\exp(tk\log(12/\epsilon)-cm)$. Take an arbitrary  $x\in \R^n$ satisfying $\|x\|_2=1, {\rm supp}(x)\subset \Omega$. There is some $y_x\in {\mathcal N}_\epsilon$ such that $\|x-y_x\|\leq \epsilon$.
Note that
\begin{eqnarray*}
\min_{I\subseteq [m], \abs{I} \geq m/2}\|A_I{ x}\|_2&=&\min_{I\subseteq [m], \abs{I} \geq m/2}\|A_I{ \left((x-y_x)+y_x\right)}\|_2\\
&\geq & \min_{I\subseteq [m], \abs{I} \geq m/2} \left(\|A_Iy_x\|_2-\|A_I(x-y_x)\|_2\right) \geq \sqrt{c_-}-\sqrt{\theta_+}\epsilon.
\end{eqnarray*}
We can take $\epsilon$ small enough so that $\sqrt{c_-}-\sqrt{\theta_+}\epsilon>0$.  Set $\theta_-:=(\sqrt{c_-}-\sqrt{\theta_+}\epsilon)^2$.
Then
$$
\min_{I\subseteq [m], \abs{I} \geq m/2}\|A_I{ x}\|_2^2 \geq \theta_-
$$
holds for all $x$ with ${\rm supp}(x)\subset \Omega$ with probability at least  $1-\exp(tk\log(12/\epsilon)-cm)$. Note that there are ${n\choose tk}$ distinct $\Omega\subset [n]$ and
$$
{n\choose tk}\leq \left(\frac{en}{tk}\right)^{tk}.
$$
 Hence
$$
\min_{I\subseteq [m], \abs{I} \geq m/2}\|A_I{ x}\|_2^2 \geq \theta_-
$$
holds for all $tk$-sparse $x\in \R^n$  with probability at least
$$
1-\exp(tk(\log(en/tk)+\log(12/\epsilon))-cm).
$$
Thus the desired SRIP property holds for all $tk$-sparse vectors with probability at least $1-\exp(-cm/2)$ provided
$$
m \geq \frac{2}{c} tk \left(\log (en/tk)+\log(12/\epsilon)\right).
$$

\end{proof}

\begin{proof}[Proof of Theorem \ref{theo-3.3}]
Fix a $k$-sparse signal $x_0\in\R^n\setminus\{0\}$ and let $b=[b_1, \dots, b_m]^\top:=A x_0$.
For any $\epsilon \in \{1,-1\}^m$ set
$b_\epsilon := [\epsilon_1b_1, \dots, \epsilon_m b_m]^\top$. As before
we consider the following minimization problem:
\begin{equation}\label{eq:min1}
     \min \|x\|_1 \mhsp {\rm s.\, t. } \mhsp A x=b_\epsilon.
\end{equation}
The solution to (\ref{eq:min1}), if it exists, is denoted as $x_\epsilon$.
We claim that for any
$\epsilon \in \{1,-1\}^m$ we must have
$$
    \|x_\epsilon\|_1\geq \|x_0\|_1
$$
if $x_\epsilon$ exists (it may not exist), and the equality holds if and only if $x_\epsilon=\pm x_0$.

Assume the claim is false. Then either $\|x_\epsilon\|_1< \|x_0\|_1$ or
$\|x_\epsilon\|_1= \|x_0\|_1$ but $x_\epsilon \neq \pm x_0$.
Observe that $|A x_\epsilon| = |A x_0|$.
So $\innerp{a_j,x_\epsilon} = \pm \innerp{a_j,x_0}$
for all $j$. Let
$$
  I:=\{j:~\innerp{a_j,x_\epsilon} = \innerp{a_j,x_0}\}.
$$
Then either $\abs{I} \geq m/2$ or $\abs{I^c} \geq m/2$. We first assume that
$\abs{I}\geq m/2$. Then $A_I x_\epsilon = A_I x_0$.  Here $A_I=[a_{j}: j\in I]^\top$. However, $A$ satisfies the strong RIP
 of order $tk$ and levels $\theta_-,\theta_+$ with $t\geq \max\{\frac{1}{2\theta_--\theta_-^2}, \frac{1}{2\theta_+-\theta_+^2}\}$.
 Consequently  $A_I$ satisfies RIP of order $tk$ and $\delta_{tk}\leq \max\{{1-\theta_-}, {\theta_+-1}\}<\sqrt{1-\frac{1}{t}}$.
 Then, using  the results in \cite{CaiZhang13} (see also Section 3), we must have
\begin{equation}\label{eq:xI0}
    \argmin{x\in \R^m}\{\|x\|_1:~A_Ix=A_Ix_0\}\,\,=\,\, x_0.
\end{equation}
Note that
$x_\epsilon\in \{x: A_Ix=A_Ix_0\}$ and hence $\|x_0\|_1\leq \|x_\epsilon\|_1$.
Based on the assumption of either $\|x_\epsilon\|_1< \|x_0\|_1$ or $\|x_\epsilon\|_1= \|x_0\|_1$, we obtain
$\|x_\epsilon\|_1=\|x_0\|_1$, which implies that $x_\epsilon =x_0$, which is a contradiction. 
In the case $\abs{I^c} \geq m/2$, similar argument yields
$x_\epsilon = -x_0$, also a contradiction. We have now proved the theorem.
\end{proof}

\section{Discussion and future directions}
We introduce a strong notion of restricted isometry and show that it naturally allows to extend compressed sensing to phaseless measurements. That is, in contrast to minimizing the $\ell_1$ norm over $Ax = Ax_0$ to recover $x_0$ whenever $A$ satisfies RIP, we show that minimizing $\ell_1$ over the larger set $|Ax| = |A x_0|$ recovers $x_0$ up to global sign when $A$ satisfies SRIP. We show that a random $m \times n$ Gaussian matrix  satisfies SRIP for $m = O(k \log (n/k))$ and thus phaseless compressed sensing is possible in a mathematical sense with $O(k \log (n/k))$ measurements. The analysis in the paper also yields an erasure robust Johnson-Lindenstrauss Lemma. All the results above are over $\mathbb{R}$. The extension of these results to hold over $\mathbb{C}$ cannot follow the same line of reasoning and is the subject of future work. Finally, we pose as an open problem to find an algorithm which provably recovers with high probability a $k$-sparse vector in $\mathbb{C}^n$ from $O(k \log (n/k))$ from phaseless Gaussian measurements.
\vspace{1cm}

{\bf Acknowledgements.}
Work on this paper began during the AIM workshop ¡°Frame theory intersects geometry¡±
in Palo Alto. We thank the organizers for their kind invitation.
We are thankful to Jonathan Kelner, Miles Lopes, Yang Wang and Rachel Ward for fruitful discussions.

\end{document}